\documentclass[runningheads]{llncs}
\usepackage[utf8]{inputenc}
\usepackage[T1]{fontenc}
\usepackage{lmodern}
\usepackage{algorithm2e}
\usepackage{amsmath}
\usepackage{amsfonts}
\usepackage{amssymb}
\usepackage{amsmath}
\usepackage{framed}
\usepackage[hidelinks]{hyperref}
\newcommand{\softO}{\widetilde{O}}
\newcommand{\F}{\mathbb{F}}
\newcommand{\C}{\mathcal C}

\DeclareMathOperator{\Q}{\mathbb{Q}}
\DeclareMathOperator{\Z}{\mathbb{Z}}
\DeclareMathOperator{\Frac}{\mathrm{Frac}}
\DeclareMathOperator{\Disc}{\mathrm{Disc}}
\DeclareMathOperator{\tr}{\mathrm{tr}}
\newcounter{countproblem}

    \SetKwInOut{Input}{Input}
    \SetKwInOut{Output}{Output}

\begin{document}\date{\vspace{-5ex}}
\title{On the complexity of computing integral bases of function fields} 
\author{Simon Abelard\inst{1}}
\institute{Laboratoire d'informatique de l'{\'E}cole polytechnique (LIX, UMR 7161)\\ 
CNRS, Institut Polytechnique de Paris \\
\email{abelard@lix.polytechnique.fr} }
\maketitle
\begin{abstract}
  Let $\C$ be a plane curve given by an equation $f(x,y)=0$ with $f\in K[x][y]$
  a monic squarefree polynomial. We study the problem of computing an integral
  basis of the algebraic function field $K(\C)$ and give new complexity bounds
  for three known algorithms dealing with this problem. For each algorithm, we
  study its subroutines and, when it is possible, we modify or replace them so
  as to take advantage of faster primitives. Then, we combine complexity
  results to derive an overall complexity estimate for each algorithm. In
  particular, we modify an algorithm due to B\"ohm et al. and achieve a
  quasi-optimal runtime. 

\end{abstract}

\paragraph*{Acknowledgements.}
Part of this work was completed while the author was at the Symbolic
Computation Group of the University of Waterloo. This paper is part of a
project that has received funding from the French Agence de l'Innovation de
D\'efense. The author is grateful to Gr\'egoire Lecerf, Adrien Poteaux and
\'Eric Schost for helpful discussions and to Gr\'egoire
Lecerf for feedback on a preliminary version of this paper.

\section{Introduction}

When handling algebraic function fields, it is often helpful --if not
necessary-- to know an integral basis. Computing such bases has a wide range of
applications from symbolic integration to algorithmic number theory and applied
algebraic geometry. It is the function field analogue of a well-known and
difficult problem: computing rings of integers in number fields and, as often,
the function field version is easier: the algorithm of Zassenhaus~\cite{zas67}
described for number fields in the late 60's can indeed be turned into a
polynomial-time algorithm for function fields which was later precisely
described by Trager~\cite{trager}. 

However, there are very few complexity results going further than just stating
a polynomial runtime. Consequently, most of the existing algorithms in the
literature are compared based on their runtimes on a few examples and this
yields no consensus on which algorithm to use given an instance of the problem.
In this paper, we provide complexity bounds for three of the best-known
algorithms to compute integral bases and provide complexity bounds based on
state-of-the art results for the primitives they rely on. 

In this paper, we focus on the case of plane curves $\C$ given by equations of
the form $f(x,y)=0$ with $f\in K[x][y]$ monic in $y$ and squarefree. We set the
notation $n=\deg_y f$ and $d_x=\deg_x f$. The associated function field is
$K(\C)=\Frac\left(K(x)[y]/f(x,y)\right)$, it is an algebraic extension of
degree $n$ of $K(x)$. An element $h(x,y)$ of $K(\C)$ is integral (over $K[x]$)
if there exists a monic bivariate polynomial $P(x,y)$ such that $P(x,h(x,y))$
equals $0$ in $K(\C)$. The set of such elements forms a free $K[x]$-module of
rank $n$ and a basis of this module is called an integral basis of $K(\C)$.

Computing integral bases of algebraic function fields has applications in
symbolic integration~\cite{trager} but more generally an integral basis can be
useful to handle function fields. For instance, the algorithm of van Hoeij and
Novocin~\cite{novocin} uses such a basis to ``reduce'' the equation of function
fields and thus makes them easier to handle. The algorithm of Hess~\cite{hess}
to compute Riemann-Roch spaces is based on the assumption that integral
closures have been precomputed. This assumption is sufficient to establish a
polynomial runtime, but a more precise complexity estimate for Hess' approach
requires to assess the cost of computing integral closures as well.

\paragraph*{Our contribution.}
We provide complexity estimates for three algorithms dedicated to computing
integral bases of algebraic function fields in characteristic $0$ or greater
than $n$. To the best of our knowledge, no previous bounds were given for these
algorithms. Another approach which has received a lot of attention is the use
of Montes' algorithm. We do not tackle this approach in the present paper, a
complexity estimate has been given by Bauch in~\cite[Lemma~3.10]{bauch16} in
the case of number fields.  Using the Montes algorithm, a local integral
basis of a Dedeking domain $A$ at a prime ideal $\mathfrak{p}$ is computed in
$O\left( n^{1+\varepsilon}\delta\log
q+n^{1+\varepsilon}\delta^{2+\varepsilon}+n^{2+\varepsilon}\delta^{1+\varepsilon}\right)
$ $\mathfrak{p}$-small operations, with $\delta$ the $\mathfrak{p}$-valuation
of $\Disc(f)$ and the cardinal of $A/\mathfrak{p}$.

Our contribution is actually not limited to a complexity analysis: the
algorithms that we present have been slightly modified so that we could
establish better complexity results. We also discuss possible improvements to
van Hoeij's algorithm in some particular cases that are not uncommon in the
literature. Our main complexity results are Theorems~\ref{thm:cplxvh},
\ref{thm:cplxtrager} and~\ref{thm:bohmcplx}. Note that we count field
operations and do not take into account the coefficient growth in case of
infinite fields nor the field extensions incurred by the use of Puiseux series.
We also made the choice not to delve into probabilistic aspects: all the
algorithms presented here are ``at worst'' Las Vegas due to the use of Poteaux
and Weimann's algorithm, see for instance~\cite[Remark~3]{PoWe17}. 

We decided to give worst-case bounds and to only involve $n$ and
$\Disc(f)$ in our theorems so as to give ready-to-use results. Our proofs,
however, are meant to allow the interested reader to derive sharper bounds
involving more precise parameters such as the regularity and ramification
indices of Puiseux series.

We summarize these complexity estimates in Table~\ref{tab:cplx} in a simpler
context: we ignore the cost of factorizations and bound both $n$ and
$d_x=\deg_xf$ by $D$.  In this case, the input size is in $O(D^2)$ and output
size in $O(D^4)$. The constant $2\le\omega\le 3$ refers to a feasible exponent
for matrix multiplication, see~\cite{omega} for the smallest value currently
known. Translating the above bound, the complexity of the Montes approach is at
best in $\softO(D^5)$ but only for a computing a local integral basis at one
singularity, while the algorithm detailed in Section~\ref{sec:bohm} computes a
global integral basis for a quasi-optimal arithmetic complexity (i.e. in
$\softO(D^4)$).

\paragraph*{Organization of the paper.} 
We sequentially analyze the three algorithms: Section~\ref{sec:vh} is dedicated
to van Hoeij's algorithm~\cite{vh}, Section~\ref{sec:trager} to Trager's
algorithm~\cite{trager} and Section~\ref{sec:bohm} to an algorithm by B\"ohm et
al. introduced in~\cite{bohm}. In each section, we first give an overview of
the corresponding algorithm and insist on the parts where we perform some
modifications. The algorithms we describe are variations of the original
algorithms so we give no proof exactness and refer to the original papers in
which they were introduced. Then, we establish complexity bounds for each
algorithm by putting together results from various fields of computer algebra.
We were especially careful about how to handle linear algebra, Puiseux series
and factorization over $K\left[ [x] \right][y]$. 

\begin{table}[h]  
  \caption{\label{tab:cplx} Simplified complexity estimates for computing integral bases.}      \begin{center}     \renewcommand{\arraystretch}{1.5}         
    \begin{tabular}{|c|c|}               
      \hline 
      Algorithm & Worst-case complexity  \\      
      \hline                
      Trager's algorithm~\cite{trager} & $\softO(D^7) $  \\            
      Van Hoeij's algorithm~\cite{vh} & $\softO(D^{\omega+4})$ \\          
      B\"ohm et al.'s algorithm~\cite{bohm} & $\softO(D^4) $   \\    
      \hline       
    \end{tabular}    
  \end{center}
\end{table}

\section{Van Hoeij's algorithm}\label{sec:vh}

\subsection{Puiseux series}

We recall some basic concepts about Puiseux series and refer to~\cite{walker}
for more details. Assuming that the characteristic of $K$ is either $0$ or
$>n$, the Puiseux theorem states that $f\in K[x][y]$ has $n$ roots in the field
of Puiseux series $\cup_{e\ge 1}\overline{K}\left( (x^{1/e}) \right)$.

Following Duval~\cite{duval}, we group these roots into irreducible factors of
$f$. First, one can write $f=\prod_{i=1}^r f_i$ with each $f_i$ irreducible in
$K[ [x] ][y]$. Then, for $1\le i\le r$ we write
$f_i=\prod_{j=1}^{\varphi_i} f_{ij}$, where each $f_{ij}$ is irreducible in
$\overline{K}[ [x] ][y]$. Finally, for any $(i,j)\in\{1,\ldots
r\}\times\{1,\ldots,\varphi_i\}$ we write
\[f_{ij}=\prod_{k=0}^{e_i-1}\left(Y-S_{ij}(x^{1/{e_i}}\zeta_{e_i}^k)\right),\]
where $S_{ij}\in\overline{K}\left( (x) \right)$ and $\zeta_{e_i}$ is a
primitive $e_i$-th root of unity.

\begin{definition}
  The $n$ fractional Laurent series
  $S_{ijk}(x)=S_{ij}(x^{1/{e_i}}\zeta_{e_i}^k)$ are called the classical
  Puiseux series of $f$ above 0. The integer $e_i$ is called the ramification
  index of $S_{ijk}$.
\end{definition}

\begin{proposition}
  For a fixed $i$, the $f_{ij}$'s all have coefficients in $K_i$, a
  degree-$\varphi_i$ extension of $K$ and they are conjugated by the action of
  the associated Galois group. We have $\sum_{i=1}^{r}e_i\varphi_i=n$.
\end{proposition}

\begin{definition}~\cite[Definition~2]{PoWe17}
  A system of rational Puiseux expansions over $K$ ($K$-RPE) of $f$ above $0$
is a set $\{R_i\}_{1\le i\le r}$ such that 
\begin{itemize} 
  \item[$\bullet$] $R_i(T)=(X_i(T),Y_i(T))\in K_i((T))^2$, 
  \item[$\bullet$] $R_i(T)=(\gamma_iT^{e_i},\sum_{j=n_i}^\infty \beta_{ij}T^j)$, 
    where $n_i\in\Z$, $\gamma_i\ne 0$ and $\beta_{in_i}\ne 0$, 
    \item[$\bullet$] $f_i(X_i(T),Y_i(T))=0$, 
    \item[$\bullet$] the integer $e_i$ is minimal. 
\end{itemize}

\end{definition}
      
In the above setting, we say that $R_i$ is centered at $(X_i(0),Y_i(0))$.  We
may have $Y_i(0)=\infty$ if $n_i<0$ but this cannot happen if $f$ is monic.

\begin{definition}~\cite[Definition~3]{PoWe17}
The regularity index of a Puiseux series $S$ of $f$ with ramification index $e$
is the smallest $N\ge \min(0,ev_x(S))$ such that no other Puiseux series $S'$
have the same truncation up to exponent $N/e$. The truncation of $S$ up to its
regularity index is called the singular part of $S$.
\end{definition}

It can be shown that two Puiseux series associated to the same RPE share the
same regularity index so we can extend this notion (and the notion of singular
part) to RPE's.

\subsection{Description of van Hoeij's algorithm}
We will be looking for an integral basis of the form $p_i(x,y)/d_i(x)$, where
the $p_i$ are degree-$i$ monic polynomials in $y$. It is known that the
irreducible factors of the denominators $d_i$ are among the irreducible factors
of the discriminant with multiplicity at least 2. We can treat these factors
one by one by first looking for local integral bases at each of these factors,
i.e. bases whose denominators can only be powers of such an irreducible factor.
A global integral basis is then recovered from these local bases by CRT.

To compute a local integral basis at a fixed factor $\phi$, van Hoeij~\cite{vh}
follows the following strategy. Starting from $(1,y,\cdots, y^{n-1})$ and
updating it so that it generates a larger module, until this module is the
integral closure.  This basis is modified by multiplying it by an appropriate
triangular matrix in the following way. Let us fix a $d$, then $b_d$ must be a
linear combination of the $b_0$, $\cdots$, $b_{d-1}$ such that
$(yb_{d-1}+\sum_{i=0}^{d-1} a_ib_i)/\phi^j$ is integral with $j$ as large as
possible.

To this end, the coefficients of the linear combination are first set to be
variables and we write equations enforcing the fact that the linear combination
divided by $\phi$ has to be integral. If a solution of this system is found,
the value of $b_d$ is updated and we repeat the process so as to divide by the
largest possible power of $\phi$. Note that a solution is necessary unique
otherwise the difference of two solutions would be an integral element with
numerator of degree $d-1$, which means that the $j$ computed in the previous
step was not maximal. When there is no solution, we have reached the maximal
exponent and move on to computing $b_{d+1}$. 

For the sake of completeness, we
give a description of van Hoeij's algorithm but we refer to van Hoeij's
original paper~\cite{vh} for a proof that this algorithm is correct. This
algorithm is originally described for fields of characteristic 0 but also works
in the case of positive characteristic provided that we avoid wild ramification
(see~\cite[Section~6.2.]{vh}). To deal with this issue, we make the assumption
that we are either considering characteristic zero or greater than $n$. 

\begin{algorithm}[ht]\label{algo:vh} 
  \Input{A monic irreducible polynomial $f(y)$ over $K[x]$}
  \Output{An integral basis for $K[x,y]/\langle f\rangle$}

  $n\gets \deg_y(f)$ \;
  $S_{fac}\gets$ set of factors $P$ such that $P^2|\Disc(f)$\;

  \For{$\phi$ in $S_{fac}$}{
    Compute $\alpha$ a root of $\phi$ (possibly in extension) \;
    Compute $r_i$ the $n$ Puiseux expansions of $f$ at $\alpha$ with precision $N$ \;
    $b_0\gets 1$ \;
    \For{$d\gets 1$ to $n-1$}{
      $b_d\gets yb_{d-1}$ \;
      solutionfound $\gets$ true \;
      Let $a_0$, $\cdots$ $a_{d-1}$ be variables \;
      $a \gets (b_d+\sum_{i=0}^{d-1} a_ib_i)/(x-\alpha)$ \;
      \While{solutionfound}{
	Write the equations, i.e. the coefficients of $a(r_i)$ with negative power of $(x-\alpha)$ for any $i$ \;
	Solve this linear system in the $a_i$'s \;
	\If{no solution}{solutionfound $\gets$ false\;
	\Else{There is a unique solution $(a_i)$ in $K(\alpha)^d$ \;
	  Substitute $\alpha$ by $x$ in each $a_i$ \;
	  $b_d \gets (b_d+\sum_{i=0}^{d-1} a_ib_i)/\phi$ \;
	}}
	}
      }
    }
    From all the local bases perform CRT to deduce $B$ an integral basis \;
 
    \Return{B} \;
    \caption{Van Hoeij's algorithm~\cite{vh}}
\end{algorithm}

\subsection{Complexity analysis}

In this section, we prove the following theorem.

\begin{theorem}\label{thm:cplxvh}
Let $f(x,y)$ be a degree-$n$ monic squarefree polynomial in $y$.
Algorithm~\ref{algo:vh} returns an integral basis for the corresponding
function field and costs the factorization of $\Disc(f)$ and
$\softO(n^{\omega+2}\deg\Disc(f))$ field operations, where $2\le\omega\le 3$ is
a feasible exponent for linear algebra.
\end{theorem}
\begin{proof}

  First, we need to compute the discriminant and recover its square factors,
  which costs a factorization of a univariate polynomial of degree $\le nd_x$.

Then, we need to compute the Puiseux expansions $\eta_i$ of $f$ at one root of
each factor in $S_{fac}$, up to precision $N=\max_i\sum_{i\ne j} v(\eta_i-\eta_j)$.
Using the algorithm of Poteaux and Weimann~\cite{PoWe17}, the Puiseux
expansions are computed up to precision $N$ in $\softO(n(\delta+N))$ field
operations, where $\delta$ stands for the valuation of $\Disc(f)$. Indeed,
these expansions are computed throughout their factorization algorithm, which
runs in $\softO(n(\delta+N))$ field operations as stated
in~\cite[Theorem~3]{PoWe17}. Therefore, in theory, we will see that computing
the Puiseux expansions has a negligible cost compared to other parts of the
algorithm since $N \le n^2$. 

Another problem coming from the use of Puiseux expansions is that we have to
evaluate bivariate polynomials (the $b_i$'s) at the Puiseux expansions of $f$.
However this matter can be dealt with by keeping them in memory and updating
them along the computations. This way, for a fixed $d$ we first initialize
$b_d=yb_{d-1}$ so we just have to perform a product of Puiseux expansions at
precision $O(n^2)$ and then each time $b_d$ is updated it will amount to
performing a linear combination of Puiseux expansions. Since we fix precision
at $N\le n^2$, taking into account the denominator in the exponents of the
Puiseux series this amounts to handling polynomials of degrees $\le n^3$. Thus,
in our case, arithmetic operations on Puiseux series can be performed in
$\softO(n^3)$ field operations.

The main task in this algorithm is to solve a linear system of $c$ equations in
$d$ variables over the extension $K(\alpha)$, where $c$ is the total number of
terms of degrees $<1$ in the $n$ Puiseux expansions. Since we know the linear
system must have at most one solution, we have the lower bound $c \ge d$ but in
the worst case, each Puiseux series has $n$ terms of degrees $<1$ and so $c$
can be bounded above by $n^2$. More precisely, we can bound it by $ne$, where
$e$ is the maximum of the ramification indices of the classical Puiseux
expansions of $f$.  

In most cases, this system will be rectangular of size $c\times d$ so we solve
it in time $\softO(cd^{\omega-1})$ using~\cite[Theorem~8.6]{aecf}. This step is
actually the bottleneck for each iteration and using the bounds on $d$ and $c$
it runs in $\softO(n^{\omega+1}\deg\phi)$ field operations, since the extension
$K(\alpha)$ of $K$ has degree $\le \deg\phi$.

This process is iterated over the irreducible factors of the discriminant
appearing with multiplicity at least 2, and for $\phi$ such a factor we have to
solve at most $n+M(\phi)/2$ systems, where $M(\phi)$ is the multiplicity of
$\phi$ in $\Disc(f)$. Indeed, each time a solution to a system is found the
discriminant is divided by $\phi^2$ so that cannot happen more than $M(\phi)/2$
times, but since we need to make sure that we have no solution before
incrementing $d$ we will have to handle $n$ additional systems. Thus, for a
fixed factor $\phi$ the cost of solving the systems is bounded by $O(n\cdot
n^{\omega+1}\deg\phi+n^{\omega+1}\deg\phi M(\phi))$, where the factor
$\deg\phi$ comes from the fact that the linear systems are solved over a degree
$\deg\phi$-extension of the base field.

Thus, the complexity is in
$\softO\left(\sum_{\phi\in
  S_{fac}}n^{\omega+1}M(\phi)\deg\phi+n^{\omega+2}\sum_{\phi\in
    S_{fac}}\deg\phi\right)$.

\end{proof} 

\begin{remark}
  If the base field is a finite field $\F_q$, factoring the discriminant is done in $\softO(
  (nd_x)^{1.5}\log q+nd_x(\log q)^2 )$ bit operations~\cite{KedlayaUmans}.
\end{remark}
\begin{remark}

The above formula shows how the size of the
input is unsufficient to give an accurate estimate of the runtime of van
Hoeij's algorithm.  Indeed, in the best possible case $\# S_{fac}$, $\deg\phi$ and
$M(\phi)$ might be constant, and all the $c_{\phi,i}$'s might be equal to $d$,
leading to an overall complexity in $O(n^{\omega+2})$. In the worst possible
case however, the sum $\sum_{\phi\in S_{fac}}\deg\phi$ is equal to the degree of the
discriminant, leading to an overall complexity in
$\softO(n^{\omega+2}\deg\Disc(f))$.  \end{remark}

\subsection{An improvement in the case of low-degree singularities}

Instead of incrementally computing the $b_i$'s, it is possible to compute one
$b_k$ by solving the exact same systems, except that this time the previous
$b_i$'s may not have been computed (and are thus set to their initial value
$y^i$). The apparent drawback of this strategy is that it computes $b_k$
without exploiting previous knowledge of smaller $b_i$'s and therefore leads to
solving more systems than using the previous approach. More precisely, if we
already know $b_{k-1}$ then we have to solve $e_k-e_{k-1}+1$ systems otherwise
we may have to solve up to $e_k+1$ systems. Using the complexity analysis
above, we can bound the complexity of finding a given $b_k$ without knowing
other $b_i$'s by $\softO(n^2k^{\omega-1}(e_k+1)\deg\phi)$. 

However, we know that for a fixed $\phi$, the $b_i$'s can be taken of the form
$p_i(x,y)/\phi^{e_i}$ where the exponents are non-decreasing and bounded by
$M(\phi)$. Therefore, when $M(\phi)$ is small enough compared to $n$, it makes
sense to pick a number $k$ and compute $b_k$. If $b_k=y^k$ then we know that
$b_i=y^i$ for any $i$ smaller than $k$. If $b_k=p_k(x,y)/\phi^{M(\phi)}$ then
we know that we can take $b_i=y^{i-k}b_k$ for $i$ greater than $k$. In most
cases neither of this will happen but then we can repeat the process
recursively and pick one number between $1$ and $k-1$ and another one between
$k+1$ and $n$ and repeat.

In the extreme case where we treat $M(\phi)$ as a constant (but $\deg\phi$ is
still allowed to be as large as $\deg(\Disc(f))/2$) this approach saves a
factor $\softO(n)$ compared to the iterative approach computing the $b_i$'s one
after another. This is summarized by the following proposition.

\begin{proposition}\label{prop:lowdegvh} 
Let $f(x,y)$ be a degree-$n$ monic squarefree polynomial in $y$ such that
irreducible factors of $\Disc(f)$ only appear with exponent bounded by an
absolute constant. The above modification of van Hoeij's algorithm returns
an integral basis for the corresponding function field and costs a univariate
factorization of degree $\le nd_x$ and $\softO(n^{\omega+1}\deg\Disc(f))$ field
operations, where $\omega$ is a feasible exponent for linear algebra. 
\end{proposition}

\begin{proof}
Let us first assume that $M(\phi)=1$ : then the
problem is just to find the smallest $k$ such that $e_k=1$. Since we the
$e_i$'s are non-decreasing, we can use binary search and find this $k$ after
computing $O(\log n)$ basis elements $b_i$'s, for a total cost in
$\softO(n^{\omega+1}\deg\phi)$ and we indeed gain a quasi-linear factor
compared to the previous approach. As long as $M(\phi)$ is constant, a naive
way to get the same result is to repeat binary searches to find the smallest
$k$ such that $e_k=1$, then the smallest $k$ such that $e_k=2$ and so on.
\end{proof}

\begin{remark}
Such extreme cases are not uncommon among the examples presented in the
literature and we believe that beyond this extreme, there will be a trade-off
between this strategy and the classical one for non-constant but small
multiplicities. We do not investigate this trade-off further because finding
proper turning points should be addressed in practice as it depends both on
theory and implementation.
\end{remark}

\subsection{The case of few singularities with high multiplicities}

In the other extreme case where $M(\phi)$ is greater than $n$ our strategy will
perform worse than the original one. Therefore, two ideas seem natural to find
the $e_i$: performing larger ``jumps'' by testing values of $e_i$ which are
multiples of a fixed $\nu>1$ or even following a binary search approach on each
$e_i$. We briefly explain why these strategies do not beat the classical one.

Given a root $\alpha$ of the discriminant, and fixing a $d$ between 1 and $n$,
 it is indeed possible for any $\nu>1$ to try to find elements such that
 $\left(b_d+\sum_{i=0}^{d-1} a_ib_i\right)/(X-\alpha)^{\nu}$ is integral, thus allowing
to skip steps in the iterated updates and divisions. 

But there is a price to pay for this: the system that we will have to solve is
bigger. When dividing by $(X-\alpha)$ the number of equations is the number of
terms of the Puiseux expansions of exponent $\le 1$ which we bounded by
$ne_{max}$. When dividing by $(X-\alpha)^\nu$, however, the number of
equations is bounded by $\nu ne_{max}$. When solving our rectangular system,
recall that the complexity depends linearly on the number of equations, and
thus even though this approach reduces the number of iterations by a factor
close to $\nu$, it increases the complexity of each iteration by a factor
$\nu$.

To sum up, if we want to know whether
$\left(b_d+\sum\alpha_ib_i\right)/(X-\alpha)^{m}$ is integral then it costs the
same (up to logarithmic factors) to either repeat divisions by $(X-\alpha)$ as
in van Hoeij's algorithm, to perform repeated divisions by $(X-\alpha)^\nu$ or
even to solve a single system to directly divide by $(X-\alpha)^m$. Therefore,
this strategy does not bring any advantage over the classical strategy in the
context of van Hoeij's algorithm.

\section{Trager's algorithm}\label{sec:trager}

\subsection{A description of Trager's algorithm}
Computing an integral basis amounts to computing the integral closure of the
$K[x]$-module generated by the powers of $y$. Trager's algorithm~\cite{trager}
computes such an integral closure iteratively using the following integrality
criterion to decide when to stop. Note that there exists many similar
algorithms like Round 2 and Round 4 using various criteria for integrality. A
more precise account on these algorithms and their history is given in the final
paragraphs of~\cite[Section~2.7]{diem}.

\begin{proposition}\cite[Theorem~1]{trager} 
  Let $R$ be a principal domain ($K[X]$ in our case) and $V$ a domain that is a
  finite integral extension of $R$. Then $V$ is integrally closed if and only
  if the idealizer of every prime ideal containing the discriminant equals $V$.
\end{proposition}
\begin{proof} See~\cite{trager}. \end{proof}

More precisely, Trager's algorithm uses the following corollary to the above
proposition:

\begin{proposition}\cite[Corollary~2]{trager} 
  The module $V$ is integrally closed if and only if the idealizer of the
  radical of the discriminant equals $V$.  
\end{proposition}

Starting from any basis of integral elements generating a module $V$ the idea
is to compute $\hat{V}$ the idealizer of the radical of the product of all such
ideals in $V$. Either $\hat{V}$ is equal to $V$ and we have found an integral
basis, or $\hat{V}$ is strictly larger and we can repeat the operation. We
therefore build a chain of modules whose length has to be finite.  Indeed, the
discriminant of each $V_{i}$ has to be a strict divisor of that of $V_{i-1}$.

\begin{algorithm}[ht]\label{algo:trager}
  \Input{A degree-$n$ monic squarefree polynomial $f(y)$ over $K[x]$}
  \Output{An integral basis for $K[x,y]/\langle f\rangle $}

 $D\gets \Disc(f)$ \;
 $B\gets (1,y,\cdots,y^{n-1})$\;
 \While{true}{
 Set $V$ the $K[x]$-module generated by $B$ \;
 $Q\gets \prod P_i$, where $P_i^2\vert D$\;
 If $Q$ is a unit then return $B$\;
 Compute $J_Q(V)$ the $Q$-trace radical of $V$ \;
 Compute $\hat{V}$ the idealizer of $J_Q(V)$ \;
 Compute $M$ the change of basis matrix from $\hat{V}$ to $V$\;
 Compute $\det M$, if it is a unit then return $V$ \;
 Update $B$ by applying the change of basis \; 
 $D\gets D/(\det M)^2$ \;
 $V\gets \hat{V}$ \; }
  
  \caption{A bird's eye view of Trager's algorithm~\cite{trager}}
\end{algorithm}

\paragraph*{Computing the radical.}

Following Trager, we avoid computing the radical of the ideal generated by
$\Disc(f)$ directly. First, we note that this radical is the intersection of
the radical of the prime ideals generated by the irreducible factors of
$\Disc(f)$. Let $P$ be such a factor, we then use the fact that in
characteristic zero or greater than $n$, the radical of $\langle P\rangle$ is
exactly the so-called $P$-trace radical of $V$ (see~\cite{trager}) i.e. the set
$J_P(V)=\{ u\in V | \forall w\in V, P| \tr(uw)\}$, where the trace $\tr(w)$ is
the sum of the conjugates of a $w\in K(x)[y]$ viewed as a degree-$n$ algebraic
extension of $K(x)$. 

The reason we consider this set is that it is much easier to compute than the
radical. Note that Ford and Zassenhaus' Round 2 algorithm is designed to handle
the case where this assumption fails but we do not consider this possibility
because if it should happen it would be more suitable to use van Hoeij's
algorithm for the case of small characteristic~\cite{smallcar}. This latter
algorithm is different from the one we detailed in Section~\ref{sec:vh} but
follows the same principle, replacing Puiseux series by a criterion for
integrality based on the Frobenius endomorphism.

Finally, for $Q=\prod P_i$ we define the $Q$-trace radical of $V$ to be the
intersection of all the $J_{P_i}(V)$. Here, we further restricted the $P_i$'s
to be the irreducible factors of $\Disc(f)$ whose square still divide
$\Disc(f)$. In what follows, we summarize how $J_Q(V)$ is computed in Trager's
algorithm. Once again, we refer to~\cite{trager} for further details and proofs.

Let $M$ be the trace matrix of the module $V$, i.e. the matrix whose entries
are the $(\tr(w_iw_j))_{i,j}$, where the $w_i$'s form a basis of $V$. An element
$u$ is in the $Q$-trace radical if and only if $Mu$ is in $Q\cdot R^n$. In
Trager's original algorithm, the $Q$-trace radical is computed via a $2n\times
n$ row reduction and one $n\times n$ polynomial matrix inversion. 

We replace this step and compute a $K[x]$-module basis of the $Q$-trace radical
by using an approach due to Neiger~\cite{neiger16} instead. Indeed, given a
basis $w_i$ of the $K[x]$-module $v$, the $Q$-trace radical can be identified
to the set 
\[\left\lbrace f_1,\cdots,f_n\in K[x]^n \: \bigg\vert\: \forall 1\le j\le n,\: \sum_{i=1}^n f_i\tr(w_iw_j) =0\bmod Q(x)\right\rbrace.\]

Using~\cite[Theorem~1.4]{neiger16} with $n=m$ and the shift $s=0$, there is a
deterministic algorithm which returns a basis of the $Q$-trace radical in Popov
form for a cost of $\softO(n^\omega\deg(Q))$ field operations.

\paragraph*{Computing the idealizer.}

The idealizer of an ideal $\mathfrak{m}$ of $V$ is the set of $u\in\Frac(V)$
such that $u\mathfrak{m}\subset\mathfrak{m}$. Let $M_i$ represent the
multiplication matrix by $m_i$ with input basis $(v_1,\cdots,v_n)$ and output
basis $(m_1,\cdots,m_m)$. Then to find the elements $u$ in the idealizer we
have to find all $u\in\Frac(R)$ such that $Mu\in R^{n^2}$. Note that building
these multiplication matrices has negligible cost (in $O(n^2)$ field
operations) using the technique of~\cite{trager76}.

Following Trager, we row-reduce the matrix $M$ and consider $\hat{M}$ the top left
$n\times n$ submatrix and the elements of the idealizer are now exactly the $u$
such that $\hat{M}u \in R^n$. Thus, the columns of $\hat{M}^{-1}$ form a basis
of the idealizer. Furthermore, the transpose of $\hat{M}^{-1}$ is the change of
basis matrix from $V_i$ to $V_{i+1}$. 

\subsection{Complexity analysis}

The purpose of this section is to prove the following theorem.

\begin{theorem}\label{thm:cplxtrager}
  Consider $f$ a degree-$n$ monic squarefree polynomial in $K[x][y]$, then
  Algorithm~\ref{algo:trager} returns an integral basis for the cost of
  factoring $\Disc(f)$ and $\softO(n^5\deg\Disc(f))$ operations in $K$.
\end{theorem}
\begin{proof}
The dominant parts in this algorithm are the computations of radicals and
idealizers, which have been reduced to linear algebra operations on polynomial
matrices.  First, we have already seen how to compute the $Q$-trace radical
$J_Q(V)$ in $\softO(n^\omega\deg(Q))$ field operations using the algorithm
presented in~\cite{neiger16}. 

To compute the idealizer of $J_Q(V)$, we row-reduce a $n^2\times n$ matrix with
entries in $K[x]$ using naive Gaussian elimination. This costs a total of $O(n^4)$
operations in $K(x)$.

Then we extract the top $n\times n$ square submatrix $\hat{M}$ from this
row-reduced $n^2\times n$ matrix and invert it for $\softO(n^\omega)$
operations in $K(x)$.  The output $\hat{M}^{-1}$ of this gives a basis of a
module $\hat{V}$ such that $V\subset\hat{V}\subset\overline{V}$. 

To translate operations in $K(x)$ into operations in $K$, one can bound the
degrees of all the rational fractions encountered, however it is quite
fastidious to track degree-growth while performing the operations described
above. In fact, we exploit the nature of the problem we are dealing with. 

Our first task is to row-reduce a matrix $M$ built such that a $u=\sum_{i=1}^n
\rho_iv_i$ is in $\hat{V}$ if and only if $M(\rho_1,\ldots,\rho_n)^t\in
K[x]^n$. The $\rho_i$'s are rational fractions but their denominator divides
$Q$. Therefore, we fall back to finding solutions of
$M(\tilde{u}_1,\ldots,\tilde{u}_n)^t\in \left(Q(x)\cdot K[x]\right)^n$, where
the $\tilde{u}_i$'s are polynomials. In this case, it does no harm to reduce
the entries of the matrix $M$ modulo $Q$, however performing Gaussian
elimination will induce a degree growth that may cause us to handle polynomials
of degree up to $n\deg Q$ instead of $\deg Q$. With this bound, the naive
Gaussian elemination costs a total of $O(n^5\deg Q)$ operations in $K$.

After elimination, we retrieve a $n\times n$ matrix $\hat{M}$ whose entries
have degrees bounded by $n\deg Q$. Inverting it will cause another degree
increase by a factor at most $n$. Thus, the inversion step has cost in
$\softO(n^{\omega+2}\deg Q)$. Since $\omega\le 3$, each iteration of Trager's
algorithm has cost bounded by $O(n^5\deg Q)$. 

Now, let us assess how many iterations are necessary. Let us assume that we are
exiting step $i$ and have just computed $V_{i+1}$ from $V_i$. Let us consider
$P$ a square factor of $\Disc(V_i)$. Let $\mathfrak{m}$ be a prime ideal of
$V_i$ containing $P$. Let us consider $u\in V_{i+1}$, then by definition $uP\in
\mathfrak{m}$ because $P\in\mathfrak{m}$ and therefore $u\in
\frac{1}{P}\mathfrak{m}\subset \frac{1}{P}V_{i}$. Thus,
$V_{i+1}\subset\frac{1}{P}V_{i}$. This means that at each step $i$ we have
$\Disc(V_{i+1})=\Disc(V_i)/Q_i^2$, where $Q_i$ is the product of square factors
of $\Disc(V_i)$. Thus, the total number of iterations is at most half the
multiplicity of the largest factor of $\Disc(f)$. 

More precisely, if we assume that the irreducible factors of $\Disc(f)$
are $r$ polynomials of respective degrees $d_i$ and multiplicity $\nu_i$, then
the overall complexity of Trager's algorithm is in \[\softO\left(\sum_{i=1}^\nu
n^5\sum_{j\le r,\: \nu_j\ge 2i} d_j \right),\]
where $\nu=\lfloor \max \nu_i/2\rfloor$.

Since $\sum_{i=1}^r
\nu_id_i \le \deg\Disc(f)$, the above bound is in $\softO(n^{5}\deg\Disc(f))$, which ranges between
$\softO(n^{6}d_x)$ and $\softO(n^5)$ depending on the input
$f$.
\end{proof}

\begin{remark}
In the above proof, our consideration of degree growth seems quite pessimistic
given that the change of basis matrix has prescribed determinant. It would be
appealing to perform all the computations modulo $Q$ but it is unclear to us
whether the algorithm remains valid. However, even assuming that it is
possible, our complexity estimate would become $\softO(n^4\deg\Disc(f))$, which
is still no better than the bound we give in next section.
\end{remark}

\section{Integral bases through Weierstrass factorization and truncations of Puiseux series}\label{sec:bohm}
 
As van Hoeij's algorithm, this algorithm due to B\"ohm et al.~\cite{bohm}
relies on computing local integral bases at each ``problematic'' singularity and
then recovering a global integral basis. But this algorithm then splits the
problem again into computing contributions to the integral basis at each
branch of each singularity. 

More precisely, given a reduced Noetherian ring $A$ we denote by $\overline{A}$
its normalization i.e. the integral closure of $A$ in its fraction field
$\Frac(A)$. In order to compute the normalization of $A=K[x,y]/\langle
f(x,y)\rangle$ we use the following result to perform the task locally at all
the singularities.

\begin{proposition}\cite[Proposition~3.1]{bohm}
  Let $A$ be a reduced Noetherian ring with a finite singular locus
  $\{P_1,\ldots, P_s\}$. For $1\le i\le s$, let an intermediate ring $A\subset
  A^{(i)}\subset\overline{A}$ be given such that
  $A^{(i)}_{P_i}=\overline{A_{P_i}}$. Then $\sum_{i=1}^s
  A^{(i)}=\overline{A}$.
\end{proposition}
\begin{proof} See the proof of~\cite[Proposition~3.2]{bohm2013}.\end{proof}

Each of these intermediate rings is respectively called a \emph{local
contribution} to $\overline{A}$ at $P_i$. In the case where
$A^{(i)}_{P_j}=A_{P_j}$ for any $j\ne i$, we say that $A^{(i)}$ is a
\emph{minimal local contribution} to $\overline{A}$ at $P_i$. Here, we consider
the case $A=K[x,y]/\langle f(x,y)\rangle$ and will compute minimal local
contributions at each singularity of $f$. This is summarized in
Algorithm~\ref{algo:bohm}.

\begin{algorithm}[ht]\label{algo:bohm}
  \Input{A monic irreducible polynomial $f(y)$ over $K[x]$}
  \Output{An integral basis for $K[x,y]/\langle f\rangle$}

  $n\gets \deg_y(f)$ \;
  $S_{fac}\gets$ set of factors $\phi$ such that $\phi^2|\Disc(f)$\;

  \For{$\phi$ in $S_{fac}$}{
    Compute $\alpha$ a root of $\phi$ (possibly in extension) \;
    Apply a linear transform to fall back to the case of a singularity at $x=0$ \;
    Compute the maximal integrality exponent $E(f)$ \;
    Using Proposition~\ref{prop:facto}, factor $f$ over $K[[x]][y]$\;
    Compute the B\'ezout relations of Proposition~\ref{prop:facsplit} \;
    Compute integral bases for each factor as in Section~\ref{sec:normbranch} \;
    As in Section~\ref{sec:branchsplit}, recover the local contribution corresponding to $\phi$ \;
    (For this, use Proposition~\ref{prop:facsplit} and Proposition~\ref{prop:invertfactor})
  }
    From all the local contributions, use CRT to deduce $B$ an integral basis \;
    \Return{B} \;
    \caption{Adaptation of the algorithm by B\"ohm et al.~\cite{bohm}}
\end{algorithm}

In this section, we revisit the algorithm presented by B\"ohm et al.
in~\cite{bohm} and replace some of its subroutines in order to give complexity
bounds for their approach. Note that these modifications are performed solely
for the sake of complexity and rely on algorithms for which implementations may
not be available. However, we note that our new description makes this
algorithm both simpler and more efficient because we avoid using Hensel lifting
to compute the $E(f)$ and the triples $(a_i,b_i,c_i)$ which are actually
obtained as byproducts of the factorization of $f$ over $K[ [x] ][y]$. This
allows us to prove the following theorem.

\begin{theorem}\label{thm:bohmcplx}
Let $f(x,y)$ be a degree-$n$ monic squarefree polynomial in $y$. Then
Algorithm~\ref{algo:bohm} returns an integral basis of $K[x,y]/\langle
f\rangle$ and costs a univariate factorization of degree $\deg \Disc(f)$ over
$K$, at most $n$ factorizations of degree-$n$ polynomials over an extension of
$K$ of degree $\le \deg(\Disc(f))$ and $\softO(n^2\deg\Disc(f))$ operations in
$K$. 

\end{theorem}

\subsection{Computing normalization at one branch}\label{sec:normbranch}

Let us first address the particular case when $f(x,y)$ is an irreducible
Weierstrass polynomial. This way, we will be able to compute integral bases for
each branches at a given singularity. The next section will then show how to
glue this information first into a local integral basis and then a global
integral basis can be computed using CRT as in van Hoeij's algorithm. The main
result of this section is the following proposition. 

\begin{proposition}\label{prop:brch} 
  Let $g$ be an irreducible Weierstrass polynomial of degree $m$ whose Puiseux
  expansions have already been computed up to sufficiently large precision
  $\rho$. An integral basis for the normalization of $K[[x]][y]/\langle
  g\rangle$ can be computed in $\softO(\rho m^2)$ operations in $K$. 
\end{proposition}

As in van Hoeij's algorithm, the idea is to compute for any $1\le d < m$ a
polynomial $p_d\in K[x][y]$ and an integer $e_d$ such that $p_d(x,y)/x^{e_d}$
is integral and $e_d$ is maximal. However, the building process is quite
different. We clarify this notion of maximality in the following definition.

\begin{definition}\label{def:dmax}
  Let $P\in K[x][y]$ be a degree-$d$ monic polynomial (in $y$). We say that $P$
  is $d$-maximal if there exists an exponent $e_d$ such that $P(x,y)/x^{e_d}$
  is integral and such that there is no degree-$d$ monic polynomial $Q$
  satisfying $Q(x,y)/x^{e_d+1}$.
\end{definition}

\begin{remark} We introduce the notion of $d$-maximality for the sake of
clarity and brevity. To the best of our knowledge this notion has not received
a standard name in the literature and was often referred to using the word
maximal.\end{remark}

Let us consider the $m$ Puiseux expansions $\gamma_i$ of $g$. Since $g$ is
irreducible, these expansions are conjugated but let us first make a stronger
assumption : there exists a $t\in\Q$ such that all the terms of degree lower
than $t$ of the expansions $\gamma_i$ are equal and the terms of degree $t$ are
conjugate. We truncate all these series by ignoring all terms of degree greater
or equal to $t$. This way, all the expansions share the same truncation
$\overline{\gamma}$.

\begin{lemma}\label{lem:intexp}\cite[Lemma~7.5]{bohm}
  Using the notation and hypotheses of previous paragraph, for any $1\le d < m$
  the polynomial $p_d=(y-\overline{\gamma})^d$ is $d$-maximal.
\end{lemma}
\begin{proof} See \cite{bohm}. \end{proof}

In a more general setting, more truncations are iteratively performed so as to
fall back in the previous case. We recall below the strategy followed
in~\cite{bohm} for the sake of completeness.

Initially we have $g_0=g=\prod_{i=1}^m (y-\gamma_i)$. We compute the smallest
exponent $t$ such that the expansions $\gamma_i$ are pairwise different. We
truncate the expansions to retain only the exponents smaller than $t$ and
denote these truncations $\gamma_j^{(1)}$. Among these expansions, we extract a
set of $r$ mutually distinct expansions which we denote by $\eta_i$. Note that
by local irreducibility, each of these expansions correspond to the same number
of identical $\gamma_j^{(1)}$. We further denote $\overline{g_0}=\prod_{i=1}^m
(y-\gamma_i^{(1)})$ and $g_1=\prod_{i=1}^r (y-\eta_i)$ and $u_1=m/r$. We
actually have $\overline{g_0}=g_1^{u_1}$.

We recursively repeat the operation: starting from a polynomial
$g_{j-1}\prod_{i=1}^{r_{i-1}} (y-\eta_i)$, we look for the first exponent such
that all the truncations of the $\eta_i$ are pairwise different. Truncating
these expansions up to exponent strictly smaller, we compute
$\overline{g_{j-1}}=\prod_{i=1}^{m_j} (y-\gamma_i^{(j)})$. Once again we retain
only one expansion per set of identical truncations and we define a
$g_j=\prod_{i=1}^{r_j} (y-\eta_i)$ and $u_j=m_j/r_j$.

The numerators of the integral basis that the algorithm shall return are
products of these $g_i$'s. Speaking very loosely, the $g_i$ have decreasing
degrees in $y$ and decreasing valuations so for a fixed $d$ the denominator
$p_d$ is chosen of the form $\prod g_i^{\nu_i}$ where the $\nu_i$'s are
incrementally built as follow : $\nu_1$ is the largest integer such that
$\deg_y(g_1^{\nu_1}) \le d$ and $\nu_1\le u_1$, then $\nu_2$ is the largest
integer such that $\deg_y(g_1^{\nu_1}g_2^{\nu_2})\le d$ and $\nu_2\le u_2$, and
so on. This is Algorithm 6 of~\cite{bohm}, we refer to the proof
of~\cite[Lemma~7.8]{bohm} for a proof of exactness.

Since we assumed that we are treating a singularity at $0$, the denominators
are powers of $x$. The proper exponents are deduced in the following way: for
each $g_i$ we keep in memory the set of expansions that appear, we denote this
set by $N_{g_i}$. Then for any $\gamma$ in the set $\Gamma$ of all Puiseux
expansions of $g$ we compute $\sigma_i=\sum_{\eta\in N_{g_i}}v(\gamma-\eta)$
which does not depend on the choice of $\gamma\in\Gamma$. For any $j$, if
$p_j=\prod_k g_k^{\nu_k}$ then the exponent $e_j$ of the denominator is given
by $\left\lfloor\sum_k \nu_k\sigma_k\right\rfloor$. Further justifications of
this are given in~\cite{bohm}.

\paragraph*{Complexity analysis.}

Let us now give a proof of Proposition~\ref{prop:brch}. To do so, remark that
the $g_k$'s are polynomials whose Puiseux series are precisely the truncation
$\eta_i$'s of the above $\gamma_j^{(i)}$. Equivalently, one can say that the
$g_k$'s are the norms of the Puiseux expansions $\eta_i$'s.

To compute them, we can appeal to the Algorithm NormRPE of Poteaux and
Weimann~\cite[Section~4.1.]{PoWe17}. Suppose we know all the expansions
involved up to precision $\rho$ sufficiently large. These expansions are not
centered at $(0,\infty)$ because $g$ is monic. Therefore, the hypotheses
of~\cite[Lemma~8]{PoWe17} are satisfied and the algorithm NormRPE compute each
of the $g_i$'s above in time $\softO(\rho\deg_y(g_i)^2)$.

Then we remark that the total number of such $g_i$'s is in $O(\log m)$. Indeed,
at each step the number of expansions to consider is at least halved (Puiseux
expansions are grouped according to their truncations being the same, at least
two series having the same truncation). Since the degree of each $g_i$ is no
greater than $m-1$, all these polynomials can be computed in $\softO(m^2\rho)$
operations in $K$.

Once the $g_i$'s are known we can deduce the numerators $p_i$'s as explained
above. Building them incrementally starting from $p_1$ each $p_i$ is either
equal to a $g_j$ or can be expressed as one product of quantities that were
already computed (either a $g_j$ or a $p_k$ for $k<i$). Therefore, computing
all the numerators amounts to computing at most $m$ products of polynomials
whose degrees are bounded by $m$ over $K[x]/\langle x^\rho\rangle$. Using
Sch\"onhage-Strassen's algorithm for these products the total cost is in
$\softO(\rho m^2)$ operations in $K$. The computation of denominators then has
a negligible cost. This concludes the proof.

\subsection{Branch-wise splitting for integral bases}\label{sec:branchsplit}

Once again, let us assume that we are treating the local contribution at the
singularity $x=0$. In the setting of van Hoeij's algorithm, this corresponds to
dealing with a single irreducible factor of the discriminant. We further divide
the problem by considering the factorization $f=f_0\prod_{i=1}^rf_i$, where
$f_0$ is a unit in $K[ [x] ][y]$ and the other $f_i$'s are irreducible
Weierstrass polynomials in $K[ [x] ][y]$. 

We can apply the results from the previous section to each $f_i$ for $i>0$ in
order to compute an integral basis of $K[ [x] ][y]/\langle f_i\rangle$. In this
section, we deal with two problems: we explain how to compute the factorization
of $f$ and how to efficiently perform an analogue of the Chinese Remainder
Theorem to compute an integral basis of $K[ [x] ][y]/\langle f_1\cdots
f_r\rangle$ from the integral bases at each branch. For the sake of
completeness, we recall in Section~\ref{sec:f0}, how B\"ohm et al. take $f_0$
into account and deduce a minimal local contribution at any given singularity.

\begin{proposition}\label{prop:facsplit}\cite[Proposition~5.9]{bohm}
  Let $f_1$,\ldots,$f_r$ be the irreducible Weierstrass polynomials in $K[ [x]
][y]$ appearing in the factorization of $f$ into branches. Let us set
$h_i=\prod_{j=1, j\ne i}f_j$. Then the $f_i$ and $h_i$ are coprime in $K(
(x))[y]$ so that there are polynomials $a_i, b_i$ in $K[ [x] ][y]$ and positive
integers $c_i$ such that $a_if_i+b_ih_i=x^{c_i}$ for any $1\le i\le r$.

  Furthermore, the normalization of $K[ [x] ][y]/(f_1\cdots f_r)$ splits as

  \[ \overline{K[ [x] ][y]/\langle f_1\cdots f_r\rangle} \cong \bigoplus_{i=1}^r\overline{K[ [x] ][y]/\langle f_i\rangle} \]
  and the splitting is given explicitly by

  \[(t_1\bmod f_1, \ldots,t_r\bmod f_r)\mapsto \sum_{i=1}^r\frac{b_ih_it_i}{x^{c_i}}\bmod f_1\cdots f_r.\]
\end{proposition}
\begin{proof} See~\cite[Theorem~1.5.20]{de2013local}. \end{proof}

The following corollary will be used in practice to recover an integral basis for $\overline{K[ [x] ][y]/\langle f_1\cdots f_r\rangle}$.

\begin{proposition}\label{prop:gluebases}\cite[Corollary~5.10]{bohm}
  With the same notation, let
  \[\left(1,\frac{p_1^{(i)}(x,y)}{x^{e_1^{(i)}}},\ldots,\frac{p^{(i)}_{m_i-1}(x,y)}{x^{e^{(i)}_{m_i-1}}}\right)\]
  represent an integral basis for $f_i$, where each $p^{(i)}_j\in K[x][y]$ is a
  monic degree-$j$ polynomial in $y$. For $1\le i\le r$, set
  \[\mathcal{B}^{(i)}=\left(\frac{b_ih_i}{x^{c_i}},\frac{b_ih_ip_1^{(i)}}{x^{c_i+e^{(i)}_1}},\ldots,\frac{b_ih_ip_{m_i-1}^{(i)}}{x^{c_i+e^{(i)}_{m_i-1}}}\right)
  .\] Then $\mathcal{B}^{(1)} \cup \cdots\cup\mathcal{B}^{(r)}$ is an integral
  basis for $f_1\cdots f_r$.
\end{proposition}

In~\cite{bohm}, these results are not used straightforwardly because the authors
remarked that it was time-consuming in practice. Instead, the $c_i$'s are
computed from the singular parts of the Puiseux expansions of $f$ and
polynomials $\beta_i$ replace the $b_i$'s, playing a similar role but
being easier to compute. 

Indeed, these $\beta_i$'s are computed in~\cite[Algorithm~8]{bohm} and they are
actually products of the polynomials $g_i$'s already computed
by~\cite[Algorithm~7]{bohm}, which is the algorithm that we detailed above to
describe the computation of an integral basis for each branch. The only new
thing to compute in order to deduce the $\beta_i$'s are the suitable exponents
of the $g_i$'s. This is achieved through solving linear congruence equations.
This step can be fast on examples considered in practice and we also note
that the $\beta_i$'s seem more convenient to handle because they are in
$K[x][y]$ and they contain less monomials than the $b_i$'s. However the
complexity of this problem (often denoted LCON in the literature) has been
widely studied, see for example~\cite{arvind2005,de2012} but, to the best of
our knowledge, none of the results obtained provide bounds that we could use
here.

For the sake of complexity bounds, we therefore suggest another way which is
based on computing the $b_i$'s of Proposition~\ref{prop:facsplit}. We also
compute the factorization of $f$ into branches in a different way: instead of
following the algorithms of~\cite[Section~7.3 \& 7.4]{bohm} we make direct
use of the factorization algorithm of Poteaux and Weimann~\cite{PoWe17} so we
also invoke their complexity result~\cite[Theorem~3]{PoWe17} which is recalled
below. Another advantage to this is that we will see that the $b_i$'s can
actually be computed using a subroutine involved the factorization algorithm,
which simplifies even further the complexity analysis.

\begin{proposition}\label{prop:facto}\cite[Theorem~3]{PoWe17}
There exists an algorithm that computes the irreducible factors of $f$
in $K[ [x] ][y]$ with precision $N$ in an expected $\softO(\deg_y(f)(\delta+N))$
operations in $K$ plus the cost of one univariate factorization of degree at
most $\deg_y(f)$, where $\delta$ stands for the valuation of $\Disc(f)$.
\end{proposition}
\begin{proof} See~\cite[Section~7]{PoWe17}. \end{proof}

Let us now get back to the first steps of Algorithm~\ref{algo:bohm}: we have to
compute $E(f)$ to assess up to what precision we should compute the Puiseux
series and then compute the factorization of $f$, the integers $c_i$ and the
polynomials $b_i$.

In each section, we tried to keep the notation of the original papers as much
as we could which is why we introduced $E(f)$ but the definition given
in~\cite[Section~4.8]{bohm} is exactly the same as the $N$ in van Hoeij's
paper~\cite{vh}. This bound can be directly computed from the singular part of
the Puiseux expansions of $f$. We recall its definition: $E(f)=\max_i\sum_{i\ne
j} v(\gamma_i-\gamma_j)$, where the $\gamma_i$'s are the Puiseux expansions of
$f$. We will see later on an alternate definition which will make it easier to
bound $E(f)$.

Following~\cite{bohm}, we need to compute the factorization of $f$ into
branches up to precision $E(f)+c_i$. Using Poteaux and Weimann's factorization
algorithm from Proposition~\ref{prop:facto}, we can compute the factors $f_i$
up to the desired precision. 

Furthermore, using a subroutine contained within this algorithm, we can compute
the B\'ezout relation $a_if_i+b_ih_i=x^{c_i}$ up to precision $E(f)+c_i$. This
is detailed in~\cite[Section~4.2]{PoWe17}, where our $c_i$ is the lifting
order $\kappa$ and our $f_i$ and $h_i$ are respectively the $H$ and $G$ of
Poteaux and Weimann. The algorithm used to compute the B\'ezout relations is
due to Moroz and Schost~\cite{MoSc} and its complexity is given
by~\cite[Corollary~1]{MoSc}. 

\paragraph*{Complexity analysis.}

We analyze the cost of the computations performed in this section and summarize
them by the following proposition.

\begin{proposition}\label{prop:splitcplx}
  Let $f(x,y)$ be a degree-$n$ monic squarefree polynomial in $y$ and let
  $\delta$ be the $x$-valuation of $\Disc(f)$. Then the integers $c_i$'s and
  $E(f)$, a factorization in branches $f=f_0\prod_{i=1}^rf_i$ as well as the
  polynomials $a_i$'s and $b_i$'s or Proposition~\ref{prop:facsplit} can be
  computed up to precision $E(f)+c_i$ for a univariate factorization degree $n$
  over $K$ and a total of $\softO(n^2\delta)$ field operations.
\end{proposition}

\begin{proof}

  First, the singular parts of the Puiseux series of $f$ above $0$ are computed
  for $\softO(n\delta)$ field operations by~\cite[Theorem~1]{PoWe17}.
  This allows us to compute $E(f)$.

Then we compute the factorization in branch up to a sufficient precision to
compute the $c_i$'s. We then extend the precision further so as to compute the
factorization and the B\'ezout relations $a_if_i+b_ih_i=x^{c_i}$ up to
precision $E(f)+c_i$.

Invoking~\cite[Corollary~1]{MoSc}, computing a single B\'ezout relation up to
precision $E(f)+c_i$ costs $\softO(n(E(f)+c_i))$ field operations. Computing
the factorization of $f$ in branches up to the same precision with
Proposition~\ref{prop:facto} accounts for $\softO(n(\delta+c_i+E(f))$
operations in $K$ and one univariate factorisation of degree $n$ over $K$. 

Using~\cite[Definition~4.14]{bohm}, we note that $E(f)$ can also be seen as
$e_{n-1}$, which is bounded by the valuation $\delta$ of the discriminant
because we assumed that we were handling a singularity at $x=0$. Thanks
to~\cite[Proposition~8]{PoWe17} we can bound $c_i$ by $v_x\left(\frac{\partial
f}{\partial y}\right)$ which is itself bounded by $\delta$. 

Putting these bounds together, the overall cost is one univariate factorization
of degree $n$ over $K$ and $\softO(n\delta)$ operations in $K$ for the
factorization step while the $n$ B\'ezout relations requires
$\softO(n^2\delta)$ operations in $K$. This concludes the proof.
\end{proof}

\subsection{Contribution of the invertible factor $f_0$}\label{sec:f0}

To deal with this problem, we reuse the following result without modification.

\begin{proposition}\label{prop:invertfactor}\cite[Proposition~6.1]{bohm}
  Let $f=f_0g$ be a factorization of $f$ with $f_0$ and $g$ in $K[ [x] ][y]$,
  $f_0$ a unit and $g$ a Weierstrass polynomial of $y$-degree $m$. Let
  $\left(p_0=1,\frac{p_1}{x^{e_1}},\cdots,\frac{p_{m-1}}{x^{e_{m-1}}}\right)$
  be an integral basis for $K[ [x] ][y]/\langle g\rangle$ such that the $p_i$'s are
  degree-$i$ monic polynomials in $K[x][y]$ and let $\overline{f_0}$ a monic
  polynomial in $K[x][y]$ such that $\overline{f_0} = f_0 \bmod x^{e_{m-1}}$.
  Let us denote $d_0=\deg_y(\overline{f_0})$.

  Then
  \[\left(1,y,\ldots,y^{d_0-1},\overline{f_0}p_0,\frac{\overline{f_0}p_1}{x^{e_1}},\cdots,\frac{\overline{f_0}p_{m-1}}{x^{e_{m-1}}}\right)\]
is an integral basis for the normalization of  $K[ [x] ][y]/\langle f\rangle$.
\end{proposition}
\begin{proof} See~\cite{bohm} \end{proof}

Since we handle a single singularity at $0$, the previous
basis is also a $K[x]$-module basis of the minimal local contribution at this
singularity by~\cite[Corollary~6.4]{bohm}.

\paragraph*{Complexity analysis.}
This step involves a truncation of $f_0$ modulo $x^{e_{m-1}}$ and $m$ products
of polynomials in $K\left[ [x] \right][y]/\langle x^{e_{m-1}} \rangle $ whose
$y$-degrees are bounded by $n=\deg_y(f)$. This incurs $\softO(mne_{m-1})$ field
operations. Since we are treating a singularity at $x=0$, we have
$e_{m-1}=O(\delta)$ with $\delta$ the valuation of $\Disc(f)$ so that we can
simplify the above bound as $\softO(n^2\delta)$ field operations.

\subsection{Proof of Theorem~\ref{thm:bohmcplx}}\label{sec:bohmcplx}

In this section, we put all the previous bounds together and prove
Theorem~\ref{thm:bohmcplx}.

\begin{proof}

As in van Hoeij's algorithm, we first compute $\Disc(f)$ and factor it in order
to recover its irreducible square factors.  For each irreducible factor $\phi$
such that $\phi^2\vert\Disc(f)$, we compute the corresponding minimal local
contribution.  For each of them, we first perform a translation so as to handle
a singularity at $x=0$. If there are several conjugated singularities we can
handle them like in van Hoeij's algorithm, at the
price of a degree-$\deg\phi$ extension of $K$ which we denote by $K'$ in this
proof. Also note that through this transform the multiplicity $M(\phi)$
corresponds to the valuation $\delta$ of the discriminant.

First, we split $f$ into branches using Proposition~\ref{prop:splitcplx} for a
cost in $\softO(n^2M(\phi))$ operations in $K'$ and one univariate
factorization of degree $\le n$ over $K'$. 

Then, at each branch $f_i$, we apply Proposition~\ref{prop:brch} with precision
$\rho=E(f)+c_i$. Therefore, the cost of computing an integral basis at each
branch $f_i$ is in $\softO(M(\phi) \deg_y(f_i)^2)$ operations in $K'$. Since
$\sum_i\deg(f_i)\le n$, computing the integral bases at all the branches costs
$\softO(n^2M(\phi))$ operations in $K'$. 

At the end of this step, we have integral bases $\mathcal{B}_i$ of the form
\[\left(1,\frac{p_1(x,y)}{x^{e_1}},\ldots,\frac{p_{m_i-1}(x,y)}{x^{e_{m_i-1}}}\right)\]
with $m_i=\deg_yf_i$ but the $p_i$'s are in $K'[ [x]][y]$. 

At first glance, this is a problem because Proposition~\ref{prop:gluebases}
requires the $p_i$'s to be in $K'[x][y]$. However, the power of $x$ in the
denominators is bounded a priori by $E:=E(f)+\max_{1\le i\le r}c_i$ so we can
truncate all series beyond this exponent. Indeed, forgetting the higher order
terms amounts to subtracting each element of the basis by a polynomial in
$K'[x]$. Such polynomials are obviously integral elements so they change nothing
concerning integrality.

We can thus apply Proposition~\ref{prop:gluebases} to get an integral basis for
$f_1\cdots f_r$. This costs $O(n)$ operations in $K'\left[ [x]
\right][y]/\langle x^E,f(x,y)\rangle$. Each such operation amounts to
$nE$ operations in $K'$. We have previously seen that $E$ is in
$O(M(\phi))$ so the overall cost of applying Proposition~\ref{prop:gluebases}
is in $O(n^2M(\phi))$ operations in $K'$.

After this process, the basis that we obtained must be put in ``triangular
form'' (i.e. each numerator $p_i$ should have degree-$i$ in $y$ in order for us
to apply Proposition~\ref{prop:invertfactor}. To do this, we first reduce every
power of $y$ greater or equal to $n$ using the equation $f(x,y)=0$. For a
fixed $i$, by the B\'ezout relations, $h_i$ has $y$-degree $\le n-m_i$ and
$b_i$ has $y$-degree $<m_i$, so we have to reduce a total of $O(n)$ bivariate
polynomials whose degrees in $y$ are in $O(n)$. Using a fast Euclidean
algorithm, this amounts to $\softO(n^2)$ operations in $K'[x]/\langle
x^E\rangle$, hence a cost in $O(n^2M(\phi))$ operations in $K'$. 

Once done, every element in the basis can be represented by a vector of
polynomials in $K'[x]$ whose degrees are bounded by $E$. To put the above
integral basis in triangular form, it suffices to compute a Hermite Normal Form
of a full rank $n\times n$ polynomial matrix. Using~\cite[Theorem~1.2]{labahn} an algorithm by Labahn,
Neiger and Zhou performs this task in $\softO(n^{\omega-1}M(\phi))$
operations in $K'$.

We can finally apply Proposition~\ref{prop:invertfactor} and deduce the minimal
local contribution for the factor $\phi$ in $\softO(n^2M(\phi))$ operations in
$K'$.

Overall, given a factor $\phi$, computing the corresponding minimal local
contribution to the normalization of $K[\C]$ costs the factorization of
$\Disc(f)$, one univariate factorization of degree $\le n$ over $K$ and
$\softO(n^2M(\phi))$ operations in $K'$. Computing all the local contributions
can therefore be done for the factorization of $\Disc(f)$, $\# S_{fac}$
univariate factorization of degree $\le n$ over extensions of $K$ of degree
$\le \max_{\phi\in S_{fac}}\deg\phi$  and $\softO(n^2\deg\Disc(f))$ operations
in $K$. 

In the case of conjugate singularities, we follow the idea of van Hoeij rather
than~\cite[Remark~7.17]{bohm} and simply replace $\alpha$ by $x$ in the
numerators and $(x-\alpha)$ by $\phi$ in the denominators because it does not
harm our complexity bound. In this process, some coefficients of the
numerators are multiplied by polynomials in $x$, which clearly preserves
integrality. Since the numerators are monic in $y$, no simplification can
occur and the basis property is also preserved.

Finally, a global integral basis for $K[x,y]/\langle f\rangle$ is deduced by a
Chinese remainder theorem. This can be achieved in quasi-linear time in the
size of the local bases. Each of them being in $O(n^2\deg\Disc(f))$, this last
CRT does not increase our complexity bound. This concludes the proof.

\end{proof}

\begin{remark}
  The $n$ factorizations incurred by the use of Poteaux and Weimann's algorithm
  are only necessary to ensure that quotient rings are actually fields, this
  cost can be avoided by using the D5 principle~\cite{d5} at the price of a
  potential complexity overhead. However, using directed
  evaluation~\cite{evaldir} yields the same result without hurting our
  complexity bounds.
\end{remark}

\section*{Conclusion}
In the setting of Table~\ref{tab:cplx}, the best bound given in this paper is
in $\softO(D^4)$ which is quasi-quadratic in the input size, but quasi-linear
in the output size. It is surprising that we are able to reach optimality
without even treating the local factors $f_i$ through a divide-and-conquer
approach like in~\cite{PoWe17}. This would allow us to work at precision
$\delta/n$ instead of $\delta$ most of the time, but this does not affect the
worst-case complexity of the whole algorithm. From an implementation point of
view, however, this approach will probably make a significant different.

Note that we are still relatively far from having implementations of algorithms
actually reaching these complexity bounds because we lack implementations for
the primitives involved in computing Popov/Hermite forms, Puiseux series and
factorizations over $K\left[ [x] \right][y]$. In some experiments we performed,
Puiseux series were actually the most time-consuming part, which is why
Trager's algorithm may still be a competitive choice despite our complexity
results.
\bibliographystyle{plain}
\bibliography{biblio}

\newcommand{\SortNoop}[1]{}
\begin{thebibliography}{10}

\bibitem{arvind2005}
Vikraman Arvind and TC~Vijayaraghavan.
\newblock The complexity of solving linear equations over a finite ring.
\newblock In {\em Annual Symposium on Theoretical Aspects of Computer Science},
  pages 472--484. Springer, 2005.

\bibitem{bauch16}
Jens-Dietrich Bauch.
\newblock Computation of integral bases.
\newblock {\em Journal of Number Theory}, 165:382--407, 2016.

\bibitem{bohm}
Janko B{\"o}hm, Wolfram Decker, Santiago Laplagne, and Gerhard Pfister.
\newblock Computing integral bases via localization and {H}ensel lifting.
\newblock {\em arXiv preprint arXiv:1505.05054}, 2015.

\bibitem{bohm2013}
Janko B{\"o}hm, Wolfram Decker, Santiago Laplagne, Gerhard Pfister, Andreas
  Steenpa{\ss}, and Stefan Steidel.
\newblock Parallel algorithms for normalization.
\newblock {\em Journal of Symbolic Computation}, 51:99--114, 2013.

\bibitem{aecf}
Alin Bostan, Fr{\'e}d{\'e}ric Chyzak, Marc Giusti, Romain Lebreton,
  Gr{\'e}goire Lecerf, Bruno Salvy, and {\'E}ric Schost.
\newblock {\em Algorithmes efficaces en calcul formel}.
\newblock 2017.

\bibitem{de2012}
Niel de~Beaudrap.
\newblock On the complexity of solving linear congruences and computing
  nullspaces modulo a constant.
\newblock {\em arXiv preprint arXiv:1202.3949}, 2012.

\bibitem{de2013local}
Theo De~Jong and Gerhard Pfister.
\newblock {\em Local analytic geometry: Basic theory and applications}.
\newblock Springer Science \& Business Media, 2013.

\bibitem{d5}
Jean Della~Dora, Claire Dicrescenzo, and Dominique Duval.
\newblock About a new method for computing in algebraic number fields.
\newblock In {\em European Conference on Computer Algebra}, pages 289--290.
  Springer, 1985.

\bibitem{diem}
Claus Diem.
\newblock {\em On arithmetic and the discrete logarithm problem in class groups
  of curves}.
\newblock Habilitation, Universit\"at Leipzig, 2009.

\bibitem{duval}
Dominique Duval.
\newblock Rational {P}uiseux expansions.
\newblock {\em Compositio mathematica}, 70(2):119--154, 1989.

\bibitem{hess}
Florian Hess.
\newblock Computing {R}iemann--{R}och spaces in algebraic function fields and
  related topics.
\newblock {\em Journal of Symbolic Computation}, 33(4):425--445, 2002.

\bibitem{vh}
Mark {\SortNoop{Hoeij}}van~Hoeij.
\newblock An algorithm for computing an integral basis in an algebraic function
  field.
\newblock {\em Journal of Symbolic Computation}, 18(4):353--363, 1994.

\bibitem{novocin}
Mark {\SortNoop{Hoeij}}van~Hoeij and Andrew Novocin.
\newblock A reduction algorithm for algebraic function fields.
\newblock 2008.

\bibitem{smallcar}
Mark {\SortNoop{Hoeij}}van~Hoeij and Michael Stillman.
\newblock Computing an integral basis for an algebraic function field, 2015.

\bibitem{evaldir}
Joris {\SortNoop{Hoeven}}van~der Hoeven and Gr{\'e}goire Lecerf.
\newblock {Directed evaluation}.
\newblock working paper or preprint, December 2018.

\bibitem{KedlayaUmans}
Kiran~S Kedlaya and Christopher Umans.
\newblock Fast polynomial factorization and modular composition.
\newblock {\em SIAM Journal on Computing}, 40(6):1767--1802, 2011.

\bibitem{labahn}
George Labahn, Vincent Neiger, and Wei Zhou.
\newblock Fast, deterministic computation of the {H}ermite normal form and
  determinant of a polynomial matrix.
\newblock {\em Journal of Complexity}, 42:44--71, 2017.

\bibitem{omega}
Fran{\c{c}}ois Le~Gall.
\newblock Powers of tensors and fast matrix multiplication.
\newblock In {\em Proceedings of the 39th international symposium on symbolic
  and algebraic computation}, pages 296--303, 2014.

\bibitem{MoSc}
Guillaume Moroz and {\'E}ric Schost.
\newblock A fast algorithm for computing the truncated resultant.
\newblock In {\em Proceedings of the ACM on International Symposium on Symbolic
  and Algebraic Computation}, pages 341--348, 2016.

\bibitem{neiger16}
Vincent Neiger.
\newblock Fast computation of shifted {P}opov forms of polynomial matrices via
  systems of modular polynomial equations.
\newblock In {\em Proceedings of the ACM on International Symposium on Symbolic
  and Algebraic Computation}, pages 365--372, 2016.

\bibitem{PoWe17}
Adrien Poteaux and Martin Weimann.
\newblock Computing {P}uiseux series: a fast divide and conquer algorithm.
\newblock {\em arXiv preprint arXiv:1708.09067}, 2017.

\bibitem{trager76}
Barry~Marshall Trager.
\newblock Algorithms for manipulating algebraic functions.
\newblock {\em SM thesis MIT}, 1976.

\bibitem{trager}
Barry~Marshall Trager.
\newblock {\em Integration of algebraic functions}.
\newblock PhD thesis, Massachusetts Institute of Technology, 1984.

\bibitem{walker}
Robert~J Walker.
\newblock {\em Algebraic curves}.
\newblock 1950.

\bibitem{zas67}
Hans Zassenhaus.
\newblock Ein algorithmus zur berechnung einer minimalbasis {\"u}ber gegebener
  ordnung.
\newblock In {\em Funktionalanalysis Approximationstheorie Numerische
  Mathematik}, pages 90--103. Springer, 1967.

\end{thebibliography}

\end{document}